\documentclass[11pt]{amsart}
\usepackage{amsthm,amssymb,bbm}
\usepackage{mathtools}

\newtheorem{theorem}{Theorem}
\newtheorem{lemma}[theorem]{Lemma}
\newtheorem{corollary}[theorem]{Corollary}

\numberwithin{equation}{section}

\def\const{C}
\def\tr{{\rm tr \,}}

\def\cz{\mathbb{C}}
\def\rz{\mathbb{R}}

\def\bp{{\bold p}}

\def\bx{{\bold x}}
\def\bR{{\bold R}}
\def\by{{\bold y}}

\def\cC{\mathcal{C}}
\def\cD{{\mathcal D}}
\def\cE{{\mathcal E}}
\def\cO{{\mathcal O}}

\def\cX{\mathcal X}
\def\cW{{\mathcal W}}

\def\gH{\mathfrak{H}}
\def\gQ{\mathfrak{Q}}
\def\gS{\mathfrak{S}}

\def\nz{\mathbb{N}}
\def\rz{\mathbb{R}}

\def\rd{\mathrm{d}}
\def\ri{\mathrm{i}}

\def\dhf{{D^\mathrm{HF}_\gamma}}

\newenvironment{dedication}{\itshape\center}{\par\medskip}
\newenvironment{acknowledgments}{\bigskip\small\noindent\textit{Acknowledgments.}}{\par}

\title{Relativistic Exchange Bounds} \author{Long Meng}
\address{Mathematisches Institut\\ Ludwig-Maximilians-Universit\"at
  M\"unchen\\ Theresienstr. 39\\ 80333 M\"unchen\\Germany}
\email{mltongji@gmail.com} \author{Heinz Siedentop}
\address{Mathematisches Institut\\ Ludwig-Maximilians-Universit\"at
  M\"unchen\\ Theresienstr. 39\\ 80333 M\"unchen\\Germany}
\email{h.s@lmu.de} \author{Matthias Tiefenbeck}
\address{Mathematisches Institut\\ Ludwig-Maximilians-Universit\"at
  M\"unchen\\ Theresienstr. 39\\ 80333 M\"unchen\\Germany}
\email{tiefenbeck@gmail.com}



\begin{document}
\maketitle
\begin{dedication}
  In admiration on the occasion of Bernard Helffer's 75th birthday
\end{dedication}
\begin{abstract}
  We collect estimates of the exchange energy of the relativistic
  no-pair Hartree-Fock and M\"uller functional and use them to show the
  existence of a minimizer and stability of matter of the relativistic
  M\"uller functional in the free picture.
\end{abstract}

\section{Introduction}
Hartree-Fock theory is well established in non-relativistic quantum
mechanics which agrees energetically for heavy atoms with the
underlying non-relativistic quantum mechanics up to the subleading
order in the atomic number (see Lieb and Simon \cite{LiebSimon1977}
for the leading order, Siedentop and Weikard
\cite{SiedentopWeikard1986,SiedentopWeikard1987O,SiedentopWeikard1988,SiedentopWeikard1989}
(upper and lower bound) and Hughes \cite{Hughes1986,Hughes1990} (lower
bound) for the leading correction, and Bach \cite{Bach1993} and
Fefferman and Seco
\cite{FeffermanSeco1989,FeffermanSeco1990O,FeffermanSeco1992,FeffermanSeco1993,FeffermanSeco1994}
for the subleading order). Moreover agreements of the densities are
shown on the scale $Z^{-\frac13}$ in \cite{LiebSimon1977} and on
smaller scales by Iantchenko \textit{et al.}
\cite{Iantchenkoetal1996}, Iantchenko \cite{Iantchenko1997} and by
Iantchenko et al. \cite{IantchenkoSiedentop2001} (density matrix). The
same accuracy of the energy has been shown for the M\"uller functional
\cite{Siedentop2009,Siedentop2014}.  In the course of proving those
results, estimates on the Hartree-Fock and M\"uller mean fields, the
direct parts, and the exchange energies, were established.

In relativistic quantum mechanics -- which is necessary for an
accurate quantitative description of large $Z$-atoms -- the situation
is well established for the one-particle case: the multi-center Dirac
operator with several atomic Coulomb charges is
\begin{equation}
  \label{eq:D}
  D_{c,\varphi} := c\boldsymbol\alpha\cdot \bp + \beta c^2 - c^2- \underbrace{\sum_{k=1}^K{Z_k\over|\bx-\bR_k|}}_{=:\varphi(\bx)}
\end{equation}
self-adjointly realized in $\gH:=L^2(\rz^3:\cz^4)$ in a canonical way
(see Schminke \cite{Schmincke1972}, Nenciu \cite{Nenciu1976} for
$K=1$, Klaus \cite{Klaus1980} for $K=2$, and Esteban \textit{et al.}
\cite{Estebanetal2021} for a general case that includes arbitrary
$K\in\nz$). Here we use $\boldsymbol\alpha, \beta$ for the four Dirac
matrices in canonical representation, $\bp:= -\ri\nabla$, and
$c\in\rz_+$ and $Z_1,..., Z_K\in(0,c)$ are known as the velocity of
light and the atomic numbers of the nuclei generating the Coulomb
potentials. For the free Dirac operator we also write $D_c$ instead of
$D_{c,0}$

For $K=1$ we may assume $\bR_1=0$. In addition
set $Z:=Z_1$ and write -- in abuse of notation -- $D_{c,Z}$ instead of
$D_{c,\varphi}$. In this case point spectrum was found by Darwin
\cite{Darwin1928}, Gordon \cite{Gordon1928}.  Note for
curiosity that the actual spectrum was known already to Sommerfeld
\cite{Sommerfeld1915} before the Dirac operator was written down
(Dirac \cite{Dirac1928}). The form domain of $D_{c,Z}$ is by
construction $\gQ:=H^\frac12(\rz^3:\cz^4)$. Its domain, easily
characterized as $H^1(\rz^3:\cz^4)$ for $\kappa:=Z/c\in[0,1/2]$ by
Kato-Rellich perturbation theory, was found by Morozov and M\"uller
\cite{MorozovMuller2017} for all $Z\in[0,1)$. The elements in $\gH$
are addressed as spinors which we write as complex valued functions on
$\Gamma:=\rz^3\times\{1,2,3,4\}$; the elements in $\Gamma$ are written
as $x:=(\bx,\sigma)$ or as $4$-vector valued functions on $\rz^3$ as
deemed practical.

However, the situation is less canonical for more electrons: The
non-uniqueness of the splitting of the one-particle space $\gH$ into
electron and positron space already makes the Fock space of the
electron-positron field non-unique even in the external field case
without photons (see Thaller \cite[Section 10]{Thaller1992}). This
non-uniqueness drips down to the no-pair Hamiltonians (Sucher
\cite{Sucher1980}, Pilkuhn \cite[Section 3.7]{Pilkuhn2005}), and,
eventually, is also inherited by relativistic Hartree-Fock and related
theories.

\medskip

To the relativistic no-pair Hartree-Fock functional we start with
identifying the one-electron space by fixing a potential $V$ which in
turn gives rise to the projection
$P_V:=\chi_{(-c^2,\infty)}(D_{c,V})$. The one-particle electron space
defined by $V$ is $\gH_V:=P_V\gH$.  We set
\begin{equation}
  \label{eq:G}
  G:=\{\gamma|T_c^\frac12\gamma T_c^\frac12\in\gS^1(\gH), 0\leq \gamma\}\
  \mathrm{and}\ G^N:=\{\gamma\in G| \tr\gamma\leq N\}.
\end{equation}
(We write $\gS^p(\gH)$ and $\|\cdot\|_p$ for the Schatten ideal of
order $p$ over the Hilbert space $\gH$ and its norm. Since
$\|\cdot\|_\infty$ is equal to the operator norm $\|\cdot\|$, we drop
the index ad libitum.) Note that $|\bp|\gamma$ might not be a trace
class operator for $\gamma\in G$. Nevertheless, we write as a
shorthand in abuse of notation $\tr(|\bp|\gamma)$ instead of
$\tr(|\bp|^\frac12\gamma|\bp|^\frac12)$ and also do in similar cases
for other operators.

The no-pair relativistic molecular Hartree-Fock functional is
\begin{equation}
  \label{nopairhf}
  \begin{split}
    \cE^\mathrm{HF}_{c,Z}:& G \to \rz\\
    \gamma \mapsto& \tr[(D_{c,\varphi}\gamma) +\cD[\rho_\gamma] - \cX[\gamma]
    + \underbrace{\sum_{1\leq k<l\leq K}{Z_kZ_l\over|\bR_k-\bR_l|}}_{=:R}
  \end{split}
\end{equation}
where
\begin{equation}
  \label{eq:DX}
  \cD(\rho,\sigma) :=\frac12\int_{\rz^3}\rd \bx \int_{\rz^3}\rd \by
  {\overline{\rho(\bx)}\sigma(\by)\over|\bx-\by|}, \
  \cX(\gamma,\delta):= \frac12\int_\Gamma\rd \bx \int_\Gamma\rd \by
  {\overline{\gamma(x,y)}\delta(x,y)\over|\bx-\by|}.
\end{equation}
The associated quadratic forms $\cD[\rho]$ and $\cX[\gamma]$ are
known as classical interaction energy of the charge distribution
$\rho$ and the exchange energy of the reduced one-particle density
matrix $\gamma$. Note that we write $\gamma(x,y)$ for the integral
kernel of $\gamma$ and
$\rho_\gamma(\bx) :=\sum_{\sigma=1}^4\lambda_n |\xi_n(\bx,\sigma)|^2$
where the $\xi_n$ is an orthonormal system of eigenspinors of $\gamma$
with associated eigenvalues $\lambda_n$, i.e.,
\begin{equation}
  \label{eigenfunktionen}
  \gamma=\sum_n\lambda_n|\xi_n\rangle\langle\xi_n|.
\end{equation}
In the language of quantum chemistry the $\xi_n$ are called natural
orbitals and the $\lambda_n$ occupation numbers of the $\xi_n$.

We note that in the atomic case, i.e., $K=1$, the last term in
\eqref{nopairhf} is not present.

At first sight $\cE^\mathrm{HF}_{c,Z}$ looks like the nonrelativistic
Hartree-Fock functional with the Laplace operator replaced by the free
Dirac operator on functions with $4$ spin components. However, because
of mathematical and physical reasons it needs to be restricted, namely
to
\begin{equation}
  G_V:=\{\gamma\in G|\gamma\leq P_V\}
\end{equation}
or
\begin{equation}
  \ G_{V,N}:=\{\gamma\in G_V|\tr\gamma\leq N\}
\end{equation}
with $T_c:=\sqrt{c^2\bp^2+c^4}-c^2$. The one-electron space $\gH_V$
fixed by $V$ defines what Sucher \cite{Sucher1980} calls
``picture''. The operator $V$ reflects physically the exterior
potential and the mean field of the electrons. Therefore reasonable
assumptions on $V$ are:
\begin{enumerate}
\item   $(D_{c,V}+\ri)^{-1}-(D_{c,0}+\ri)^{-1}\in\gS^\infty(\gH)$
\item The interval $(-2c^2,-c^2)$ is in the resolvent set of
$D_{c,\epsilon V}$ for all $\epsilon\in[0,1]$
\item $\gQ_V:=P_V\gQ\subset\gQ$. 
\end{enumerate}

The Euler equations of the functional
$\cE^\mathrm{HF}_{c,Z}\Big|_{\gQ_V}$ are the no-pair Hartree-Fock
equations, namely
\begin{equation}
  [P_V\dhf P_V,\gamma]=0,\ 0\leq \gamma\leq P_V,\ \tr\gamma= Z,\ T_c^{1/2}\gamma T_c^{1/2}\in \gS^1(\gH)
\end{equation}
where
\begin{equation}
  \label{eq:DHF}
  \dhf:= D_{c,\varphi}+\phi_\gamma -X_\gamma
\end{equation}
where for any given $\gamma\in G$ the screening potential
$\phi_\gamma$ and the exchange operator $X_\gamma$ are defined by
\begin{equation}
  \label{eq:austausch}
\phi_\gamma:=\rho_\gamma*|\cdot|^{-1}\ \mathrm{and}\ X_\gamma(x,y):= {\gamma(x,y)\over|\bx-\by|}.
\end{equation}

We note in passing that for $N\in\nz$, the minimizer can be picked as
a projection, i.e., $\gamma^2=\gamma$, by a minor modification of
Bach's argument \cite{Bach1992}.

A solution of the non-uniqueness problem induced by $V$ was offered
by Mittleman \cite{Mittleman1981} by choosing the most stable electron
space, i.e., by maximizing over $V$. The resulting stationary
condition becomes
\begin{equation}
  \label{eq:df}
  [\dhf,\gamma]=0
\end{equation}
with $\gamma^2=\gamma\leq \chi_{(-c^2,\infty)}(\dhf)$,
$V=-\varphi+\phi_\gamma-X_\gamma$,
$\gamma,|\bp|^{1/2}\gamma |\bp|^{1/2} \in \gS^1(\gH)$, and
$ \tr\gamma=Z$. This equation is known as Dirac--Fock equation (see,
e.g., \cite{EstebanSere1999,Paturel2000,Sere2023}) or relativistic
Hartree-Fock equation. The relationship between Dirac--Fock model and
Mittleman's definition is studied in
\cite{Barbarouxetal2005,Meng2023}.

\section{Estimates on  exchange terms in relativistic Hartree-Fock theory}

We start with a well known observation (see, e.g., Lieb and Yau
\cite[Eq. (7.4)]{LiebYau1988}):
\begin{lemma}
  \label{lemma1}
  We have the pointwise bound
\begin{equation}
  \label{gr}
  \sum_{\sigma,\tau=1}^4|\gamma(x,y)| \leq \sqrt{\rho_\gamma(\bx)}\sqrt{\rho_\gamma(\by)}.
\end{equation}
\end{lemma}
\begin{proof}
  The claim follows immediately from the eigenfunction expansion
  \eqref{eigenfunktionen} of $\gamma$ using the Schwarz inequality in
  the summation index $n$.
\end{proof}

Next, we address the relative form compactness of
$W_\gamma$.
\begin{lemma}
  \label{lemma2}
  Suppose $\gamma\in G$. Then
  \begin{equation}
    \label{eq:rfc}
    \sqrt{\phi_\gamma}(\bp^2+1)^{-\frac14}\in \gS^8(\gH).
  \end{equation}
  In particular, $\phi_\gamma$ is relatively form compact with respect to $(\bp^2+1)^{1/2}$.
\end{lemma}
\begin{proof}
We recall Hoffmann-Ostenhofs' inequality for $|\bp|$
\cite[Theorem 7.13]{LiebLoss2001}
\begin{equation}
  \label{eq:ho2}
 \int_{\rz^3} ||\bp|^{\frac12}\sqrt{\rho_\gamma} (x)|^2\leq  \tr(|\bp|\gamma)
\end{equation}
Thus, using first the Sobolev inequality and then  \eqref{eq:ho2} we get
\begin{equation}
  \label{eq:32}
  \left(\int \rho_\gamma^\frac32\right)^\frac23
  \leq C (\sqrt{\rho_\gamma},|\bp|\sqrt{\rho_\gamma})
  \leq \tr(|\bp|\gamma)<\infty.
\end{equation}
Therefore $\rho_\gamma\in L^\frac32(\rz^3)$ if $\gamma\in G$.

From the Hardy-Littlewood-Sobolev inequality, it follows (see Lieb and
Loss \cite[Section 4.3, Eq. (9)]{LiebLoss2001})
\begin{equation}
  \label{eq:hls}
  \|\phi_\gamma\|_{L^r(\mathbb{R}^3)}\leq C\|\rho_\gamma\|_p
\end{equation}
for
$$\frac1p=\frac23+\frac1r,\ p\in(1,\tfrac32).$$
Thus $\phi_\gamma\in L^r(\rz^3)$ for all $r\in(3,\infty)$. Using
Seiler's and Simon's inequality \cite{SeilerSimon1975} we get
$\sqrt{\phi_\gamma(x)}(\bp^2+1)^{-\frac14}\in\gS^8(\gH)$ and thus
relative compactness.
\end{proof}

Next we show that the exchange operator is a Hilbert-Schmidt operator
and estimate its norm:
\begin{lemma}
  \label{XHS}
  Assume $\gamma\in G$. Then $X_\gamma\in \gS^2(\gH)$ and
  \begin{equation}
    \label{eq:hsn}
    \|X_\gamma\|_2\lesssim \tr(|\bp|\gamma).
  \end{equation}
\end{lemma}
\begin{proof}
Pick $\xi\in H^{\frac12}(L^2(\rz^3:C^4)$. Then, using the
Hardy-Littlewood-Sobolev inequality followed by the Sobolev inequality we get
\begin{equation}
  \int_{\Gamma}\rd x\int_{\Gamma}\rd y \left|{\xi(x)^*\xi(y)\over|\bx-\by|}\right|^2
\lesssim \||\xi|^2\|_\frac32^2 \lesssim (\xi,|\bp|\xi)^2.
\end{equation}
Taking the root, replacing $\xi$ by $\xi_n$, multiplying by
$\lambda_n$ and summing over $n$ we get by the eigenfunction expansion
\eqref{eigenfunktionen} of $\gamma$
\begin{equation}
  \label{eq:hs}
  \|X_\gamma\|_2\lesssim \sum_n\lambda_n  (\xi_n,|\bp|\xi_n)\lesssim
  \tr(|\bp|\gamma) <\infty.
\end{equation}
Thus $X_\gamma$ is a Hilbert-Schmidt operator and therefore trivially
relatively form compact with respect to $| D_c|$.
\end{proof}

Next we turn to some quantitative operator bounds. The minimum of the
no-pair Hartree-Fock functional is supposed to approximate the atomic
ground-state energy which -- as long as $N\geq \const Z$ and $Z/c <1$
-- is expected to be $O(Z^\frac73)$, the same as the Thomas-Fermi
energy, and also the energy of the test functions proving the upper
bound of the Scott correction in Frank \textit{et al.}
\cite{Franketal2008,Franketal2009} and
\cite{HandrekSiedentop2015}. Thus, it is reasonable to consider only
those $\gamma\in G_Z$ that fulfill
\begin{equation}\label{eff1}
  \tr\left(T_c \gamma\right)=\mathcal{O}(Z^\frac73).
\end{equation}

\begin{lemma}
  \label{lemma4}
  Assume $\gamma\in G^N$ with $N=O(Z)$ and
  $\tr(T_c\gamma)= O(Z^\frac73)$. Then
  $\|X_\gamma\|,\|\phi_\gamma\|=O(Z^\frac53)$.
\end{lemma}
\begin{proof}
By (operator) positivity of the Coulomb kernel we have for all
$f\in\gH$ using Schwarz again
\begin{equation}
  0\leq (f,X_\gamma f)
  \leq \int_{\Gamma}\rd x |f(x)|^2 \int_{\rz^3}\rd \by{\rho_\gamma(\by)\over|\bx-\by|}
  =(f,\phi_\gamma f).
\end{equation}
Thus it suffices to show the bound for $\phi_\gamma$.

Now, for every $\bx\in\rz^3$
\begin{align}
  \label{schranke}
  \phi_\gamma(\bx) &\leq \sum_{n} {\lambda_n\over dc} \int_{\Gamma}\rd \by {|\xi_n(y)|^2\over|\bx-\by|} 
  \notag\\
  &\leq \sum_n{\lambda_n\over dc}(\xi_n, (| D_{c}|-c^2+
  {\pi^2\over4}d^2  c^2)\xi_n) \notag\\
 &\leq {1\over
    dc}\tr\left((| D_{c}|-c^2+ \tfrac{\pi^2}8 d^2c^2)\gamma\right)
  \lesssim Z^\frac53
\end{align}
where we use $\sqrt{1-t}\leq 1-t/2$ for $t\leq1$ and the translational invariance of the differential operator and Herbst's massive version of Kato's inequality \cite[Eq. (2.1)]{Herbst1977}
\begin{equation}
  \label{eq:herbst}
  {d\over|\cdot|} \leq  \sqrt{\bp^2+1}-\sqrt{1-\left(\tfrac\pi2 d\right)^2}
\end{equation}
for all $d\in(0,2/\pi)$. We use it after rescaling $x\to cx$,
$p\to \bf{p}/c$ and setting $d:=c^{-\frac13}$. (In passing we note
that the bound can be improved for $d$ close to $2/\pi$: it stays
uniformly positive (see Raynal \textit{et al.}  \cite{Raynaletal1994} and Tix \cite{Tix1997,Tix1998}). However, this is not needed here.)
\end{proof}

Lemmata \ref{lemma2}, \ref{XHS}, and \ref{lemma4} have an immediate
consequence:
\begin{corollary}
  \label{WHFbounded}
  Suppose $\gamma\in G_Z$ and $\tr(T_c\gamma)=O(Z^\frac73)$. Then the
  total Hartree-Fock mean-field potential
  $W_\gamma:=\phi_\gamma-X_\gamma\in \cW_{\frac13}$ is relatively form
  compact with respect to $\sqrt{\bp^2+1}$ and
  $\|W_\gamma\|=O(Z^\frac53)$.
\end{corollary}
To further discuss the exchange energy the following bound of the
massless kinetic energy in terms of the massive one is useful.
\begin{lemma}
  \label{lemmap}
  For $\lambda\in(0,1)$ and $c\in\rz_+$ we have
  \begin{equation}
    \label{eq:p}
    |\bp| \leq {1\over c\lambda}\left(\sqrt{c^2|\bp|^2+c^4}-c^2+ c^2\lambda^2\right).
  \end{equation}
\end{lemma}

  \begin{proof}
    It suffices to show \eqref{eq:p} for $c=1$, since the general case
    is obtained by scaling $\bp\to c\bp$.
    
    We compute the infimum of quotient of the right and left side:
    \begin{equation}
      f_\lambda(t):= {\sqrt{t^2+1}-\sqrt{1-\lambda^2}\over t}
  \end{equation}
  which is assumed at $t_m=\lambda/\sqrt{1-\lambda^2}$ with
  $f_\lambda(t_m)=\lambda$ which shows
  \begin{equation}
    \label{eq:lambda}
    |\bp|\leq \lambda^{-1}(\sqrt{|\bp|^2+1}-\sqrt{1-\lambda^2})
    \leq \lambda^{-1}(\sqrt{|\bp|^2+1}-1+\lambda^2)
  \end{equation}
where we use $\lambda\in(0,1)$ and therefore $1-\lambda\in(0,1)$.
\end{proof}

Eventually we mention that the exchange energy of the Hartree-Fock
minimizer is of lower order:
\begin{lemma}
  \label{HFaustausch}
  Assume $\gamma\in\gS^1(\gH)$ with $\tr\gamma=O(Z)$,
  $\gamma^2=\gamma$, and $\tr(T_c\gamma)=O(Z^\frac73)$. Then
  \begin{equation}
    \label{X}
    \cX[\gamma]=O(Z^\frac53).
  \end{equation}
\end{lemma}
\begin{proof}
  By the correlation inequality of Mancas \textit{et al.} ~\cite[Theorem
  1]{Mancasetal2004} (setting $N=2$) and by the estimate
  \begin{equation}
    \label{eq:exmaximal}
    \int_{|\bx-\by|<R(\bx)}
    \rd \by {\sigma(\by)\over|\bx-\by|}\geq {\sqrt[3]{9\pi}\over2}(M\sigma)(\bx)^\frac13\ \mathrm{with}\
    \int_{|\bx-\by|<R(\bx)}\rd \by\; \sigma(\by) = R(\by)
  \end{equation}
  (see \cite[Lemma 2]{Mancasetal2004}) on the exchange hole we have
  \begin{equation}
    \label{eq:austauschloch}
    {1\over|\bx-\by|} \geq \sigma*|\cdot|^{-1}(\bx)+ \sigma*|\cdot|^{-1}(\by)
    - {\sqrt[3]{9\pi}\over2} (M\sigma)(\bx)^\frac13
    - {\sqrt[3]{9\pi}\over2} (M\sigma)(\by)^\frac13
    -D[\sigma]
  \end{equation}
  where $\sigma$ is any non-negative function with
  $\int_{\rz^3}\sigma>1/2$ and $D[\sigma]<\infty$ and $M(\sigma)$ is
  the maximal function of $\sigma$. We integrate the inequality
  against $(\rho_\gamma(\bx)\rho_\gamma(\by)-|\gamma(x,y)|^2)/2$ and get
  \begin{multline}
      D[\rho_\gamma]- \frac12\int_{\Gamma\times\Gamma}\rd x \rd y
    {|\gamma(x,y)|^2\over|\bx-\by|}\\
    \geq 2D(\rho_\gamma,\sigma)N -
    2D(\rho_{\gamma},\sigma)- \sqrt[3]{9\pi} \int_{\rz^3}\rd \bx
    (M\sigma)(x)^\frac13\rho(\bx) -\frac12 D[\sigma](N^2-N).
\end{multline}
  Picking
  $$\sigma:= {2+\sqrt2\over N} \rho_\gamma$$
  yields
  \begin{multline}
    \label{Mancas}
    \int_\Gamma\rd x \int_\Gamma\rd y {|\gamma(x,y)|^2\over|\bx-\by|}\lesssim {D[\rho_\gamma]\over N} + \|M\rho_\gamma\|^\frac14_\frac43\|\rho\|_\frac43\\
    \lesssim  {\tr(|\bp|\gamma)N\over N} +\int_{\rz^3}\rho_\gamma^\frac43 \lesssim \tr(|\bp|\gamma) \leq  {\tr T_c\gamma\over \lambda c} + \lambda cN
      =O(Z^\frac53)
  \end{multline}
  where we use Kato's inequality, the maximal inequality, Daubechies'
  inequality \cite[Inequality (3.4)]{Daubechies1983}, and \eqref{eq:p}
  with $\lambda= Z^{-\frac13}$.
\end{proof}
Using the same argument as in the last line of \eqref{Mancas} on
\eqref{eq:hs} yields
\begin{corollary}
  Assume $\gamma\in G$ with $\tr(T_c\gamma)=\cO(Z^\frac73)$. Then
  \begin{equation}
    \cX[\gamma]= \cO(Z^\frac53).
  \end{equation}
\end{corollary}

\section{Estimate on exchange energy in relativistic M\"uller theory}

We will consider the M\"uller functional. The basic inequality is the
bound on the ground state energy of of an hydrogen atom by Tix
\cite[Eq. (2.1)]{Tix1998}. Inserting it in the Ritz variational
principle gives the inequality
\begin{equation}
  \label{eq:tix}
  {1\over c\kappa_k} T_c+c
  \geq {1\over|\bx|}
\end{equation}
in $\gH_{0}$ with $\kappa_k:=2/(\pi/2+2/\pi)$.
Next we,
seemingly unmotivated, introduce the no-pair operator of two
gravitating particles in the free picture. Set
\begin{equation}
  \label{eq:muller}
  \begin{split}
  \cC_{\varphi,\lambda_0}[\psi]&:\gQ_0\otimes \gQ_0\to \rz\\
  \psi&\mapsto\tfrac12\langle\psi|[(T_c-\varphi)\otimes 1+1\otimes(T_c-\varphi)-|\bx-\by|^{-1}-\lambda_0]  \psi \rangle
  \end{split}
\end{equation}
with $\lambda_0 \in \rz$. The form $\cC_{\varphi,\lambda_0}$ is
bounded from below for $0\leq Z+1/2 \leq c\kappa_k$ defining the
operator $C_{_\varphi,\lambda_0}$ according to Friedrichs. (See the first step of
Theorem \ref{existence} for more details on the boundedness). Moreover,
\begin{equation}
  \label{eq:ess}
  \sigma_\mathrm{ess}(C_{\varphi,\lambda_0}) = [0,\infty) 
\end{equation}
if we pick
\begin{equation}
  \label{37}
  \lambda_0:=\inf \sigma(C_{0,0}),
\end{equation}
and assume $Z\geq1/2$. This is a special case of Morozov's HVZ theorem
\cite[Theorem 3.1.1]{Morozov2008M} or \cite[Theorem
6]{Morozov2008}. (A slightly more restricted version of the HVZ
theorem not covering the case of attracting particles was already
given by Jakuba\ss a-Amundsen \cite{JakubassaAmundsen2005}.) We will
fix this choice \eqref{37} of $\lambda_0$ and write simply $C_\varphi$
instead of $C_{\varphi,\lambda_0}$.

The relativistic M\"uller functional is the same as the Hartree-Fock
functional except $\gamma$ is replaced by the operator square root
$\gamma^\frac12$ in the exchange term and the renormalizing term
$-\lambda_0\tr\gamma$ is added, i.e., in the molecular case it reads
\begin{equation}
  \label{nopairm}
  \begin{split}
    \cE^\mathrm{M}_{c,\varphi}:& G\to \rz\\
    \gamma \mapsto& \tr\left[(D_{c,\varphi}-\lambda_0)\gamma\right] +\cD[\rho_\gamma] -
    \cX[\gamma^{\frac12}] +\sum_{1\leq k<l\leq
      K}{Z_kZ_l\over|\bR_k-\bR_l|}
  \end{split}
\end{equation}
restricted to $G_{\varphi}$. The renormalizing term makes the infimum
of the M\"uller functional for $\varphi=0$ equal to zero.  This
renormalization of the energy is not needed for the Hartree-Fock
functional but is similar to the non-relativistic setting (Frank et
al.  \cite{Franketal2007}). The same phenomena happens already for the
simple Thomas-Fermi-Dirac functional (Lieb \cite[Section
VI.a]{Lieb1981}).

We have the following estimate on the M\"uller exchange energy in
terms of the relativistic kinetic energy:
\begin{theorem}
  \label{Mulleraustausch}
  Assume $\gamma \in G$. Then for all $d\in(0,2/\pi)$ and $c>0$
  \begin{equation}
    \label{eq:m1}
    \cX[\gamma^\frac12] = {1\over dc}\tr\left[\left(T_c +{\pi^2\over4} d^2c^2\right)\gamma\right].
  \end{equation}
  Furthermore, if $\tr (T_c\gamma) = \cO(Z^\frac73)$, then
  \begin{equation}
     \cX[\gamma^\frac12] =O(Z^\frac53).
  \end{equation}
\end{theorem}
\begin{proof}
  For fixed $\by\in\rz^3$ we set $g_y(x):=\gamma^\frac12(x+y,y)$. Then
  \begin{align}
    \label{eq:m2}
    \cX[\gamma^\frac12]& = \frac12 \int_{\Gamma}\rd y
    \int_{\Gamma}\rd x{|g_y(x)|^2\over|\bx|} \notag\\
    &\leq {1\over 2dc} \int_{\Gamma}\rd y
    (g_y,\left(T_c+ \left({\pi^2 d^2 c^2\over4}\right)g_y\right)
   = {1\over dc}\tr\left[\left(T_c+ {\pi^2 d^2 c^2\over4}\right)\gamma\right]
  \end{align}
  where use again Herbst's massive Kato inequality similar to the
  second estimate in \eqref{schranke}.  This proves the first claim.

  Picking $d:= Z^{-\frac13}$ shows the second claim.
\end{proof}

\begin{theorem}
  \label{Mulleraustausch-free=picture}
  Assume $\gamma \in G_0$. Then
  \begin{equation}
    \label{eq:m1f}
    \cX[\gamma^\frac12] = \frac12\tr\left[\left({1\over c\kappa_K}T_c +c\right)\gamma\right].
  \end{equation}
  \end{theorem}
\begin{proof}
  We start as in the proof of Theorem \ref{Mulleraustausch}: For fixed
  $\by\in\rz^3$ we set $g_y(x):=\gamma^\frac12(x+y,y)$. (Note that
  although $\gH_0$ is not invariant under dilations it is
  invariant under translations.) Then by \eqref{eq:tix} dividing by
  $2$ and using the second assumption in the last step, we get
  \begin{align}
    \label{eq:m2a}
    \cX[\gamma^\frac12] &= \frac12 \int_{\Gamma}\rd y
    \int_{\Gamma}\rd x{|g_y(x)|^2\over|\bx|} \notag\\
    &\leq \frac12\int_{\Gamma}\rd y
    \left(g_y,\left({1\over\kappa_kc}T_c+ c\right)g_y\right)
   = \frac12\tr\left[\left({1\over c\kappa_k }T_c+  c\right)\gamma\right].
  \end{align}
This ends the proof.
\end{proof}

\section{Existence of a groundstate of the atomic M\"uller functional in the free picture}

The existence of a ground-state for the Chandrasekhar kinetic energy
was proven by \cite{Tiefenbeck2010} and also addressed by Chen
\cite{Chen2013}. We will follow the strategy of \cite{Tiefenbeck2010}
but will extend it to the free picture in the diction of Sucher and to
the molecular case.

Since $\gamma\in\gS^1(\gH_0)$ implies that
$\gamma^\frac12\in\gS^2(\gH_0)$ its integral kernel
$u(x,y):=\gamma^\frac12(x,y)$ is in $\gQ_0\otimes\gQ_0$.  This
allows to recast the M\"uller functional as
\begin{equation}
  \label{eq:D2}
  \cE^{\rm M}_{c,\varphi}(\gamma)
  =\frac{1}{2}\left( u,C_\varphi - \lambda_0)u\right)_{\gH_0\otimes \gH_0}
  +\cD[\rho_\gamma] +R.
\end{equation}

\begin{theorem}
  \label{existence}
  Let $(Z+\tfrac{1}{2})<\kappa_kc$ and $Z\geq1/2$. Then
  $\cE^\mathrm{M}_{c,\varphi}$ has a minimizer $\gamma_*\in
  G_{0,N}$.
\end{theorem}
\begin{proof}
We split the proof into the following steps.

{\bf Step 1. Some properties of minimizing sequences.}  To study
minimizing sequences, we first bound the M\"uller functional from
below.

Note that for any $\gamma\in G_{\varphi,N}$,
$D_{c,0}\gamma= T_c\gamma$. Since $\cD$ is a positive quadratic form
(Onsager's lemma), we know that
\begin{equation}\label{eq:lowbound}
  \begin{split}
  \cE^{\rm M}_{c,\varphi}(\gamma) & \geq  \tr(T_c\gamma)
 - \left(\frac\kappa{\kappa_c} +\frac1{2\kappa_kc}\right)\tr (T_c\gamma)
        -[\lambda_0+ c(Z+\tfrac12)]\tr\gamma\\
                                  &=\left[ 1-{1\over\kappa_kc}\left(Z+\tfrac12\right)\right]\tr T_c\gamma -
                                    (\lambda_0+cZ+\tfrac12)N.
  \end{split}
\end{equation}

As $0\in G_{0,N}$, we know $E_{c,\varphi}^{\rm M}(G_{0,N})\leq
0$. Thus, there exists a minimizing sequence $\gamma_n$ in $G_{0,N}$
such that $\cE^{\rm M}_{c,\varphi}[\gamma_n]\leq 0$. From
\eqref{eq:lowbound}, we infer
\begin{equation}
  \tr(T_c\gamma_n) \leq
  { (\lambda_0+cZ+\tfrac12)N\over 1-{Z+\frac12\over\kappa_kc}}.
\end{equation}
This implies that $T_c^\frac12\gamma_n T_c^\frac12$ is uniformly
bounded in $\mathfrak{S}^1(\gH_0)$. Thus, up to subsequences, there
exists a density matrix $\gamma_*$ such that
$T_c^\frac12\gamma_* T_c^\frac12\in \mathfrak{S}^1(\gH_0)$ and
\begin{align}\label{eq:weak-converg}
    T_c^\frac12\gamma_n T_c^\frac12 \rightharpoonup  T_c^\frac12\gamma_* T_c^\frac12\qquad \textrm{in }\mathfrak{S}^1(\gH_0).
\end{align}

\medskip 

Obviously $\gamma*$ is a candidate for a minimizer. To show that this
is true, we are going to use \eqref{eq:D2} and the kernel of the
operator $\gamma_n^\frac12$. Next, we are going to study the properties of
$\gamma_n^\frac12$. Since
\begin{equation}
  \label{44}
    \|T_c^\frac12\gamma_n^\frac12\|^2_{\mathfrak{S}^2(\gH_0)}=\tr(T_c^\frac12\gamma_nT_c^\frac12) =\tr(T_c\gamma_n)
  \end{equation}
  the sequence $T_c^\frac12\gamma_N^\frac12$ is bounded in
  $\gS^2(\gH_0)$ and therefore -- possibly extracting a subsequence
  which in abuse of notation is written again as $\gamma_n$ -- is also
  weakly convergent to some $T_c^\frac12 \gamma_{**}T_c^\frac12$ for
  some $\gamma_{**}$ which, however, again by \eqref{44}, equals
  $\gamma_*$ since $T_c>0$, i.e.,
\begin{equation}
    T_c^\frac12\gamma_n^\frac12 \rightharpoonup   T_c^\frac12\gamma_*^\frac12\qquad \textrm{in }\mathfrak{S}^2(\gH_0).
\end{equation}
We write $u_n$ and $ u_*$, both in $\gH_0\otimes
\gH_0$, for the kernels of $\gamma_n^\frac12$ and
$\gamma_*^\frac12$ respectively. Then
$(T_c^\frac12\otimes1)u_n$ and $(T_c^\frac12\otimes1)
u_*$ are the integral kernels of $ T_c^\frac12\gamma_n^\frac12$ and $
T_c^\frac12\gamma_*^\frac12$ respectively, and
\begin{equation}
\|(T_c^\frac12\otimes1)u_n\|_{\gH_0\otimes \gH_0}=\| T_c^\frac12\gamma_n^\frac12\|_{\mathfrak{S}^2(\gH_0)}<\infty.
\end{equation}
For any operator $G\in \mathfrak{S}^2(\gH_0)$ the associated
integra. kernel $g$ is in $\gH_0\otimes \gH_0$ and we have
\begin{equation}
  \left((T_c^\frac12\otimes1)u_n,g\right)=\tr [(T_c^\frac12\otimes1)\gamma_n^\frac12G]\to \tr[(T_c^\frac12\otimes1)\gamma_*^\frac12G]=\left((T_c^\frac12\otimes1)u_*,g\right).
\end{equation}
This implies
\begin{equation}
   (T_c^\frac12\otimes1)u_n \rightharpoonup  (T_c^\frac12\otimes1)u_*\qquad\textrm{in } \gH_0\otimes \gH_0.
\end{equation}

Next we note that $\cD[\rho_{\gamma_n}]$ is bounded. Since $\cD$ is a
scalar product there is a $\gamma_{***}$ such that -- possibly by
extracting again a subsequence --
\begin{equation}
  \label{eq:DD}
  D(\sigma,\rho_{\gamma_n})\to D(\sigma,\rho_{***})
\end{equation}
for all $\sigma$ with $D[\sigma]<\infty$. That $\gamma_{***}=\gamma_*$
follows from
$D(\sigma,\rho_{\gamma_*}) =C\tr ( |\bp|^{-2}(\sigma)\circ \gamma_*)$.

{\bf Step 2: $G_{\varphi,N}$ is weakly closed.} From
\eqref{eq:weak-converg}, it is easy to see that
\begin{align*}
  \tr \gamma_*=\|\gamma_*\|_{\mathfrak{S}^1(\gH_0)} \leq \liminf_{n\to \infty}\tr\gamma_n\leq N.
\end{align*}
On the other hand, for any function $f\in \gH_0$, we have
\begin{align*}
    \left(f, \gamma_* f\right)=\lim_{n\to \infty}(f,\gamma_n f).
\end{align*}
This and $0\leq \gamma_n\leq P_0$ for any $n\geq 1$ show that
$0\leq \gamma_*\leq P_0$. Thus
\begin{align*}
    \gamma_*\in G_{0,N}.
\end{align*}

\medskip

\noindent{\bf Step 3. Passing to the limit in the energy.} Now, we are
going to show that for the minimizing sequence $\gamma_n$ constructed
above, we have
\begin{equation}
  \inf E_{c,\varphi}^{\rm M}(G_{0,N})= \lim_{n\to \infty} \cE^{\rm M}_{c,\varphi}(\gamma_n) \geq  \cE^{\rm M}_{c,\varphi}(\gamma_*)
\end{equation}
which would prove the theorem. To this end we will make use of the
recasting \eqref{eq:D2} of the M\"uller functional treating it term by
term.

{\bf Step 3.1. Passing the limit for the term
  $\left( u_n,C_\varphi u_n\right)_{\gH_0\otimes\gH_0}$.} We split
$C_\varphi$ into its positive and negative parts $C^+\geq0$ and
$C^-\leq0$, i.e.,
\begin{equation}
 C_\varphi = C^++C^-
\end{equation}
Since $(u,C^+u)$ is a nonnegative quadratic form the Schwarz inequality implies 
\begin{equation}
  \label{+}
   \lim\inf_{n\to\infty}(u_n,C^+ u_n) \geq (u_*,C^+u_*). 
\end{equation}
By the HVZ theorem $C^-\in\gS^\infty(\gH_0)$, which turns the weak convergence into strong convergence, i.e.,
\begin{equation}
  \label{-}
  \liminf_{n\to\infty}(u_n,C^-u_n)=(u_*,C^-u_*).
\end{equation}
Combining \eqref{+} and \eqref{-} yields the wanted inequality for the
first term:
\begin{align}\label{eq:D2-converg}
     \liminf_{n\to \infty}\left( u_n,C_\varphi u_n\right)_{P_0 \gH_0 \otimes P_0\gH_0}\geq \left( u_*,C_\varphi u_*\right)_{\gH_0 \otimes\gH_0}.
\end{align}

{\bf Step 3.2. Passing to the limit for the term
  $\cD(\rho_{\gamma_n})$.} We finally consider the term
$\cD(\rho_{\gamma_n})$. Since it is a scalar product, the Schwarz
inequality implies
$\liminf_{n\to\infty}D[\rho_n]\geq D[\rho_{\gamma_*}]$.

{\bf  Step 4. Conclusion.}
Since $R$ does not depend on $\gamma$, we merely have to sum the
inequalities of the the second and third step to obtain the wanted
result.
\end{proof}

We conclude this section by remarking that we expect the
$\tr\gamma_*=N$, if $N\leq Z$ and possibly beyond although we have to
leave this question open.

\section{Stability of the second kind for the relativistic M\"uller functional}
Theorem \ref{Mulleraustausch} allows us to show stability of the
second kind, also known as stability of matter, for the M\"uller
functional in the free picture ($W=\varphi$):
\begin{theorem}
  \label{Stabilitaet}
  Assume that $0<Z_1/c,...,Z_K/c\leq \kappa< 2/\pi$,
  $\bR_1,...,\bR_K\in\rz^3$ pairwise different, and
  $\gamma\in G_{\varphi,N}$ with $N\in \nz$. Then
  \begin{equation}
    \cE^\mathrm{M}_{c,\varphi} [\gamma] \geq   -\left(c^2 + {1\over 2-\kappa\pi}-\lambda_0\right)N.
  \end{equation}
\end{theorem}
\begin{proof} We start by reminding of the inequality
\begin{multline}
  \label{LY}
  \sum_{n=1}^N \left(c|\bp_n| - \sum_{k=1}^K{Z_k\over|\bx_n-\bR_k|}\right)
  +\sum_{1\leq m<n\leq N}{1\over|\bx_m-\bx_n|}+\sum_{1\leq k<l\leq K} {Z_kZ_l\over|\bR_k-\bR_l|} \geq 0
\end{multline}
for $0\leq Z_1,...,Z_K\leq \frac2\pi c$ on $\gQ^N$ (Lieb and Yau
\cite{LiebYau1988}).

Thus rescaling \eqref{LY} by $\bx\to \frac{\kappa\pi}{2}\bx$ with a
Slater determinant with orbitals $\xi_1,...,\xi_n$ we obtain for
$\gamma = \sum_{n=1}^N|\xi_n\rangle\langle \xi_n|$ and therefore
according to Lieb \cite{Lieb1981V,Lieb1981E} or Bach \cite{Bach1992} --
adapted to the relativistic case -- for general
$\gamma\in G_{0,N}$ that
\begin{equation} \label{eq:stabhf}
    \tr[(\frac{\kappa\pi}{2}c|\bp|-\varphi(\bx))\gamma) +\cD[\rho_\gamma] - \cX[\gamma]
                    + \sum_{1\leq k<l\leq K}{Z_kZ_l\over|\bR_k-\bR_l|}\geq 0 .
\end{equation}
Using \eqref{eq:stabhf} and \eqref{eq:m2} we get
\begin{align*}
  \cE_{c,\varphi}^\mathrm{M}[\gamma] &= \cE_{c,\varphi}^\mathrm{HF}[\gamma]
  +\cX[\gamma]-\cX[\gamma^\frac12]-\lambda_0\tr\gamma\\
  &\geq
  (1-\frac{\kappa\pi}{2\lambda}-\frac1{2dc})\tr(T_c\gamma) - (c^2\lambda +c d+\lambda_0)\tr\gamma\\
  &=  -\left(c^2 + {1\over 2-\kappa\pi}-\lambda_0\right)\tr\gamma
\end{align*}
where we picked $\lambda=1$ and
\begin{equation}
  d:= {1\over 2c \left(1-{\kappa\pi\over2}\right)}
\end{equation}
to obtain the last inequality.
\end{proof}

\begin{acknowledgments}
Special thanks go to Kenji Yajima who supported a stay of H.S. at
Gakushuin University, Tokyo, where parts of the work were done.
Partial support by the Deutsche Forschungsgemeinschaft (DFG, German
Research Foundation) through TRR 352 – Project-ID 470903074 is
acknowledged.
\end{acknowledgments}

\def\cprime{$'$}

      
\end{document}